\documentclass[a4paper,11pt]{article}
\usepackage[totalwidth=13.0cm,totalheight=20.0cm]{geometry}
\usepackage{amsmath,amssymb,graphicx,pst-all,amscd,amsthm,comment}
\usepackage{tikz,etex}
\usetikzlibrary{matrix,arrows,decorations.pathmorphing}

\newtheorem{rrule}{Reduction Rule}

\newtheorem{theorem}{Theorem}
\newtheorem{lemma}{Lemma}
\newtheorem{proposition}{Proposition}

\newcommand{\name}[1]{\textsc{#1}}

\newenvironment{namedefn}[3]{
\par\addvspace{0.4\baselineskip}\fbox{%
\begin{minipage}[t]{0.9\linewidth}%
\begin{tabular}{p{18mm}p{115mm}}
    \multicolumn{2}{l}{{\name{#1}}} \\
        \textsl{Input:} & {#2} \\ 
        \textsl{Question:} & {#3} \\
   \end{tabular}
\end{minipage}}\par\addvspace{0.4\baselineskip}
}

\newenvironment{parnamedefn}[4]{
\par\addvspace{0.4\baselineskip}\fbox{%
\begin{minipage}[t]{0.9\linewidth}%
\begin{tabular}{p{18mm}p{115mm}}
    \multicolumn{2}{l}{{\name{#1}}} \\
        \textsl{Input:} & {#2} \\ 
        \textsl{Parameter:} & {#3} \\ 
        \textsl{Question:} & {#4} \\
   \end{tabular}
\end{minipage}}\par\addvspace{0.4\baselineskip}
}

\title{Parameterized Complexity of $k$-Chinese Postman Problem}
\author{Gregory Gutin, Gabriele Muciaccia\\
\small  Royal Holloway, University of London\\[-3pt]
\small Egham, Surrey TW20 0EX, UK\\[-3pt]
\small \texttt{gutin@cs.rhul.ac.uk|G.Muciaccia@cs.rhul.ac.uk}
\and Anders Yeo\\
\small Singapore University of Technology and Design \\[-3pt]
\small 20 Dover Drive, Singapore 138682, and\\
\small Department of Mathematics, University of Johannesburg \\ 
\small Auckland Park, 2006 South Africa\\[-3pt]
\small \texttt{andersyeo@gmail.com}}

\begin{document}

\maketitle

\begin{abstract}
\noindent
We consider the following problem called the $k$-Chinese Postman Problem ($k$-CPP): given a connected edge-weighted graph $G$ 
and  integers $p$ and $k$, decide whether there are at least $k$ closed walks such that every edge of $G$ is contained in at least one of them
and the total weight of the edges in the walks is at most $p$? The problem $k$-CPP is NP-complete, and van Bevern et al. (to appear) and Sorge (2013) 
asked whether the $k$-CPP is fixed-parameter tractable when parameterized by $k$. We prove that the $k$-CPP is indeed fixed-parameter tractable. In fact, we prove a stronger result: the problem admits a kernel with $O(k^2\log k)$ vertices. We prove that the directed version of $k$-CPP is NP-complete and ask whether the directed version is fixed-parameter tractable when parameterized by $k$.
\end{abstract}

\section{Introduction}

Let $G=(V,E)$ be a connected graph, where each edge is assigned a nonnegative weight (a {\em weighted graph}).  Herein $n=|V|$ and $m=|E|$. A {\em closed walk} is a non-empty multiset $T=\{e_1,\dots,e_r\}$ of edges such that there exists a permutation $\sigma$ of $\{1,\ldots ,r\}$ satisfying the following: $e_{\sigma(i)}$ and $e_{\sigma(i+1)}$ share an end-vertex for every $1\le i\le r$ (where $\sigma(r+1)=\sigma(1)$). The {\sc Chinese Postman Problem} is one of the most studied and useful problems in combinatorial optimization. 


 \begin{namedefn}%
   {{\sc Chinese Postman Problem (CPP)}}%
   {A connected weighted graph $G=(V,E)$\newline
   and an integer $p$.}%
   {Is there a closed walk on $G$ such that every edge of $G$\newline
    is contained in it and the total weight of the edges\newline
    in the walk is at most $p$?}%
 \end{namedefn}
 
\noindent In this paper, we will study the following generalisation of CPP. 
 
\begin{parnamedefn}%
   {$k$-Chinese Postman Problem ($k$-CPP)}%
   {A connected weighted graph $G=(V,E)$ and\newline
   integers $p$ and $k$.}%
   {$k$}%
   {Is there a set of  $k$ closed walks such that every\newline
   edge of $G$ is contained in at least one of them\newline
   and the total weight of the edges in the walks\newline
   is at most $p$?}%
 \end{parnamedefn}
 If a vertex $v$ of $G$ is part of input and we require that each of the $k$ walks contains $v$ then this modification of $k$-CPP is polynomial-time solvable \cite{Zha1992,Per1994}.
 However, the original $k$-CPP is NP-complete; this result was proved by Thomassen \cite{Tho1997}. The following reduction from the {\sc $3$-Cycle Partitioning  Problem} is easier than Thomassen's reduction.
 In the {\sc $3$-Cycle Partitioning Problem}, given a graph $G$, we are to decide whether the edges of $G$ can be partitioned into $3$-cycles. The problem is known to be NP-complete \cite{Hol}. Let $k=m/3$ and let the weight of each edge of $G$ be 1. Observe that the solution of $(m/3)$-CPP is of weight $m$ if and only if  the edges of $G$ can be partitioned into $3$-cycles. This reduces the {\sc $3$-Cycle Partitioning Problem} into the $k$-CPP.
 
 Note that the above reduction works because undirected graphs do not contain 2-cycles and any traversal of an edge twice is forbidden as an optimal solution must traverse each edge only once. Directed graphs may contain directed 2-cycles and so the above reduction cannot be used for digraphs unless we restrict ourselves to oriented graphs, i.e., digraphs without directed 2-cycles. However, we could not find, in the literature, a proof that the {\sc Directed $3$-Cycle Partitioning Problem} is NP-complete for oriented graphs, and so we used a different proof, in Section  \ref{sec:3}, to show that the {\sc Directed $k$-CPP}, where $G$ is a directed graph, is NP-hard.

While a large number of parameterized\footnote{For background and terminology on parameterized complexity we refer the reader to the monographs~\cite{DowneyFellows99,FlumGrohe06,Niedermeier06}.} algorithmic and complexity results have been obtained for graph, hypergraph and constraint satisfaction problems, not much research has been carried out for combinatorial optimisation problem, apart form studying local search for such problems. Perhaps, the main reason is that the standard parameterizations developed for graphs, hypergraphs and constraint satisfaction (such as the value of an optimal solution or a structural parameter) are of little interest for many combinatorial optimisation problems. Recently, Niedermeier's group identified several practically useful parameters for the CPP and its generalizations, obtained a number of results and posed some open problems, see, e.g., \cite{DoMoNiWe2013,SoBeNiWe2011,SoBeNiWe2012}. 
This research was summarized in \cite{BeNiSoWe} and overviewed in \cite{Sor2013}. 

van Bevern {\em et al.} \cite{BeNiSoWe} (see Problem 3) and Sorge \cite{Sor2013} suggested to study the $k$-CPP as parameterized problem with parameter $k$ and asked whether this parameterized problem is fixed-parameter tractable, i.e., can be solved by an algorithm of running time $O(f(k)n^{O(1)})$, where $f$ is a function of $k$ only.
In Section \ref{sec:2} we prove that the $k$-CPP is indeed fixed-parameter tractable. In fact, we prove a stronger result: the problem admits a proper\footnote{The notion of a proper kernel was introduced in \cite{AbFe2006}.}
kernel with $O(k^2\log k)$ vertices. This means 
that, in polynomial time, we can either solve an instance $(G,k)$ of the $k$-CPP or obtain another instance $(G',k')$ of $k$-CPP such that $(G,k)$ is a Yes-instance if and only if $(G',k')$ is a Yes-instance, $G'$ has $O(k^2\log k)$ vertices and $k'\le k$. In fact, in our case, $k'=k$. 

In Section \ref{sec:3} we prove that the {\sc Directed $k$-CPP} is NP-hard. It is natural to ask whether the {\sc Directed $k$-CPP} parameterized by $k$ is fixed-parameter tractable. Our approach to prove that the $k$-CPP is fixed-parameter solvable does not seem to solve the {\sc Directed $k$-CPP} and the complexity of {\sc Directed $k$-CPP} remains an open problem.

\section{Kernel for $k$-CPP}\label{sec:2} 

In this section, $G=(V,E)$ is a connected weighted graph. For a solution $T=\{T_1,\dots,T_k\}$ to the $k$-CPP on $G$ ($k\ge 1$), let $G_T=(V,E_T)$, where $E_T$ is a multiset containing all edges of $E$, each as many times as it is traversed by $T_1\cup\dots\cup T_k$. Note that given $k$ closed walks which cover all the edges of a graph, their union is a closed walk covering all the edges and, therefore, it is a solution for the CPP. Hence, the following proposition holds:

\begin{proposition}\label{thm:CPPleqkCPP}
 The weight of an optimal solution for the $k$-CPP on $G$ is not smaller than the weight of an optimal solution for the CPP on $G.$
\end{proposition}

\begin{lemma}\label{lem:kdisjointcycles1}
 Let $T$ be an optimal solution for the CPP on $G$. 
If $G_T$ contains at least $k$ edge-disjoint cycles, then an optimal solution for the $k$-CPP on $G$ has the same weight as $T$. 
Furthermore if $k$ edge-disjoint cycles in $G_T$ are given, then an optimal solution for the $k$-CPP can be found in polynomial time.
\end{lemma}
\begin{proof}
Let $\cal C$ be any collection of $k$ edge-disjoint cycles in $G_T$.
 Delete all edges of $\cal C$ from $G_T$ and observe that every vertex in the remaining multigraph $G'$ is of even degree. 
Find an optimal CPP solution for every component of $G'$ and append each such solution $F$ to a cycle in $\cal C$ which has a common vertex with $F$. 
As a result, we obtain a collection $Q$ of $k$ closed walks for the $k$-CPP on $G$ of the same weight as $T$. So $Q$ is optimal by Proposition \ref{thm:CPPleqkCPP}.
\end{proof}

Let $V=V_1\cup V_2\cup V_{\ge 3}$, where $V_1$ is the set of vertices of degree $1$, $V_2$ is the set of vertices of degree $2$ and $V_{\ge 3}$ is the set of vertices of degree at least $3$. Below we will show that, in polynomial time, we can either solve $k$-CPP or bound $|V|$ from above, by bounding $|V_1|$, $|V_{\ge 3}|$ and $|V_2|$ separately.

Let $u$ be a vertex with exactly two neighbors $v$ and $w$. 
The operation of {\em bypassing} $u$ means deleting edges $uv$ and $uw$ and adding an edge $vw$ whose weight is the sum of the weights of $uv$ and $uw$. Note that the operation of bypassing may create parallel edges.

We {\bf will} need the following lemma. There, as in the rest of the section, unless stated otherwise, the logarithms are of base 2. The {\em order} of a graph is the number of its vertices.

\begin{lemma}\cite{BodThoYeo}\label{lem:degreethreeNN}
There exists a constant $c$ such that every graph $H$ with minimum degree at least $3$ and of order at least $c k\log k$ contain $k$ edge-disjoint cycles.
Such $k$ cycles can be found in polynomial time.
\end{lemma}

The fact that the $k$ cycles in Lemma \ref{lem:degreethreeNN} can be found in polynomial time is not mentioned in \cite{BodThoYeo}, but it is easy to deduce it from the greedy algorithm given in the proof of the lemma in \cite{BodThoYeo}. The greedy algorithm repeatedly chooses a shortest cycle ($k$ times) and deletes it from $H$. Since a shortest cycle can be found in polynomial time \cite{ItRo1978}, the algorithm is polynomial.

Note that the result of Lemma \ref{lem:degreethreeNN} also holds for multigraphs as if there are parallel edges then there is a cycle of length two.

In the proof of the next lemma, we will use the following well-known fact: there is an optimal solution $T$ of the CPP on $G$ which uses at most two copies of every edge of $G$ and such a solution can be found in polynomial time, see, e.g., \cite{CoCuPuSc1997}.

\begin{lemma}\label{lem:degreeonethree}
 If $|V_1|\geq k$ or $|V_{\geq 3}|\geq ck\log k+k$, where $c$ is given in Lemma \ref{lem:degreethreeNN}, then an optimal solution for the $k$-CPP on $G$ is of the same weight as an optimal solution for the CPP on $G$.
Moreover, an optimal solution for the $k$-CPP on $G$ can be obtained in polynomial time.
\end{lemma}
\begin{proof}
 Let $V_1=\{v_1,\dots,v_r\}$ with $r\geq k$ and let $w_i$ be the neighbor of $v_i$. Now, find in polynomial time an optimal solution $T$ for the CPP on $G$. In $G_T$ every vertex is of even degree, hence $G_T$ contains two copies of the edge $v_iw_i$ for every $1\le i\le r$, thus giving at least $k$ edge-disjoint $2$-cycles. We conclude by Lemma \ref{lem:kdisjointcycles1}.

 Now, assume $|V_1|\leq k$ and $|V_{\geq 3}|\geq ck\log k+k$. Remove all vertices of degree $1$ and bypass all vertices of degree $2$. The resulting multigraph contains at least $ck\log k$ vertices and has minimum degree at least three, hence by Lemma \ref{lem:degreethreeNN} it contains at least $k$ edge-disjoint cycles, which can be found in polynomial time. Therefore, in polynomial time we are able to find at least $k$ edge-disjoint cycles in $G$ and we conclude using Lemma \ref{lem:kdisjointcycles1} (note that for every optimal solution $T$ for the CPP on $G$, $G_T$ is a supergraph of $G$).
\end{proof}

By Lemma \ref{lem:degreeonethree}, we may assume that $|V_1\cup V_{\ge 3}| = O(k\log k)$, and so it remains to bound $|V_2|$. In order to do this we will use a reduction rule, but before giving it we will show the following lemma.

\begin{lemma}\label{lem:structure}
There exists an optimal solution, $T$, for the $k$-CPP on $G$, such that 
no edge in $G$, except possibly one edge of minimum weight, is used more than twice in $T$.
\end{lemma}

\begin{proof}
  Let $xy$ be an edge of minimum weight in $G$. Let $T$ be an optimal solution for the $k$-CPP on $G$, such that $xy$ is used as many times as possible.
Assume for the sake of contradiction that $uv$ is used at least three times in $T$ and $uv$ is distinct from $xy$. Observe that $G_T$ is eulerian and contains $k$ edge-disjoint cycles.

Let $T'$ be obtained from $T$ be removing two copies of $uv$ and adding two copies of $xy$. 
If $C_1$ and $C_2$ are two edge-disjoint cycles using different copies of $uv$ in $G_T$, then we note that $uvu$ is a cycle in $G_T$ and that there exists a 
cycle containing only edges from $C_1$ and $C_2$, distinct from $uv$. This implies that 
deleting two copies of $uv$ from $T$ only decreases the maximum number of edge-disjoint cycles in $G_T$ by at most one. Adding two copies of $xy$ increases the maximum number of edge-disjoint cycles by at
least one. Therefore as $G_T$ contains $k$ edge-disjoint cycles, so does $G_{T'}$. This implies that $G_{T'}$ contains a solution to the $k$-CPP of weight at most that of $T$, which is the  desired contradiction.
\end{proof}

We are now ready to give our reduction rule.

\begin{rrule}\label{rrule:longpaths}
Let $P=v_0v_1\dots v_rv_{r+1}$ be a path in $G$ such that $v_i$ is a vertex of degree $2$ for $1\le i\le r$. If $r>k$, bypass a vertex $v_i$ such that $2\le i\le r-1$. 
Choose $v_i$ in such a way that bypassing it does not change the minimum weight of an edge in $G$.
\end{rrule}

This rule is safe because in $P$ a solution to the $k$-CPP either  duplicates all the edges or it duplicates none (with the exception of a minimum-weight edge, that may be duplicated more than once), and this is true in both $G$ and the graph $G'$ obtained after an application of the rule. Therefore duplicating a contracted edge corresponds to duplicating all the edges that where contracted. In addition, duplicating the path in $G'$ creates at least $k$ edge-disjoint cycles, as it was for $G$.

From now on, we will assume that $G$ is reduced under Reduction Rule \ref{rrule:longpaths}.

Define a multigraph $H$, such that $V(H) = V_1 \cup V_{\ge 3}$ and add $r$ edges between two vertices, $u$ and $v$, in $H$ if and only if there are exactly $r$
distinct paths from $u$ to $v$ in $G$ where all internal vertices have degree two in $G$ (an edge $uv \in E$ is such a path as it has no internal vertices).

Recall that under our assumptions, $|V(H)|=O(k\log k)$. Moreover, if there are vertices $u,v$ in $H$ such that there are at least $2k$ parallel edges between them, then $G$ contains at least $k$ edge-disjoint cycles (which can be found in polynomial time), and we may apply Lemma \ref{lem:kdisjointcycles1}. Therefore we may assume that this is not the case and this ensures that $|E(H)|=O(k^2\log k)$. Finally, since $G$ is reduced under Reduction Rule \ref{rrule:longpaths}, we have $|E(G)|=O(k^3\log k)$.

However, we can show a better bound on $|E(H)|$, which leads to an improved (polynomial) kernel.

\begin{lemma}\label{lem:edgebound}  Let $k\ge 2$ and let $c_1$ be any constant.
There exists a constant $c_2$ such that every multigraph $H$ with at most $c_1k\log k$ vertices and at least $c_2k\log k$ edges contain $k$ edge-disjoint cycles. Such cycles can be found in polynomial time.
\end{lemma}
\begin{proof}
Let $c_2 = 2c_1 +4 + (2\log c_1  + 2)$. 
This implies that the following holds (as $k \geq 2$ and therefore $\log k \geq 1$):
\begin{equation}\label{eq:2}
 c_2   \geq  2c_1 +4 + \frac{2\log c_1 + 2}{\log k}    =  \frac{2 c_1 k \log k + k(4 \log k + 2\log c_1 + 2) }{k \log k }.
\end{equation}

Alon et al.\ \cite{AloHooLin} showed that a graph with average degree $d$ and $n$ vertices has a cycle of length at most $2(\log_{d-1}n)+2$. 
Note that this result also holds for multigraphs as parallel edges form cycles of length two.

Consider the following greedy algorithm used in \cite{BodThoYeo}:
repeatedly choose a shortest cycle and delete its edges from $H$. 
We will show that by the assumptions of this theorem and the value of $c_2$, it is possible to 
run this algorithm until $k$ edge-disjoint cycles have been removed and at each step the average degree $\frac{2|E(H)|}{|V(H)|}$ is at least $4$. 
To prove this claim, let us run the algorithm until either it removed $k$ cycles or the average degree dropped below $4$ and suppose that
the algorithm stopped after removing $0\le r<k$ cycles.

Note that the number of edges removed from $H$ by the algorithm is at most the following:
\begin{equation}\label{eq:1}
r (2\log_3 n + 2) \leq r(2 \log(c_1 k^2) + 2) = r(4 \log k + 2\log c_1 + 2).
\end{equation}

Then by (\ref{eq:2}), (\ref{eq:1}) and $r<k$, we have that when the algorithm stops the graph still contains at least $2c_1k\log k$ edges, which implies that the average degree is still at least $4$. This is a contradiction, which completes the proof.
\end{proof}

Using Lemma \ref{lem:edgebound} we may assume that $|E(H)|\leq c_2k\log k$ for some constant $c_2$. 
Since $G$ is reduced under Reduction Rule \ref{rrule:longpaths}, we have $|V_2|=O(k^2\log k)$, which implies the following:

\begin{theorem}
The $k$-CPP admits a kernel with $O(k^2\log k)$ vertices.
\end{theorem}

\section{NP-completeness of Directed $k$-CPP}\label{sec:3}

Our proof below is based on two facts:
\begin{itemize}
\item Deciding whether a digraph has at least $c$ arc-disjoint directed cycles, is NP-complete, see Theorem 13.3.2 in \cite{BaGu2009};
\item The weight of an optimal solution to the  $k$-CPP on a eulerian digraph $H$ equals the weight of $H$ if and only if $H$ has at least $k$ arc-disjoint directed cycles (it follows from the fact that each closed walk has a cycle).
\end{itemize}

In the following part, $d_D^+(u)$ and $d_D^-(u)$ denote the outdegree and the indegree, respectively, of a vertex $u$.

\begin{theorem}
{\sc Directed $k$-CPP} is NP-complete.
\end{theorem}

\begin{proof}
  Let $D$ be a directed graph. Now define $D'$, as follows. Let $D'$ contain $D$ as an induced subgraph and add an extra vertex $x$. Now for every vertex $u \in V(D)$ with 
$d_D^+(u) > d_D^-(u)$ add $(d_D^+(u) - d_D^-(u))$ paths of length two from $x$ to $u$ (the central vertices on the paths are new vertices). For  every vertex $u \in V(D)$ with 
$d_D^-(u) > d_D^+(u)$ add $(d_D^-(u) - d_D^+(u))$ paths of length two from $u$ to $x$. Observe that $D'$ is eulerian. 

We will now show that $D$ contains $r$ arc-disjoint directed cycles if and only if $D'$ contains $r+d_{D'}^+(x)$ arc-disjoint directed cycles. Let ${\cal C}$ be a set of $r$
arc-disjoint directed cycles in $D$ and let $D^* = D' - A({\cal C})$ be the digraph obtained from $D'$ by deleting all arcs from cycles in ${\cal C}$. 
As $D^*$ is balanced (for every vertex the indegree is equal to the outdegree) and $d_{D^*}^+(x)=d_{D'}^+(x)$ we note that there exists at least $d_{D'}^+(x)$ arc-disjoint cycles in $D^*$
(just repeatedly remove any cycle containing $x$, which can be done $d_{D'}^+(x)$ times as $D^*$  will remain balanced). Therefore there are at least $r+d_{D'}^+(x)$ arc-disjoint cycles in $D'$.

Conversely, assume that $D'$ contains $r+d_{D'}^+(x)$ arc-disjoint cycles. At most $d_{D'}^+(x)$ of these contain $x$ and thus the remaining 
$r$ cycles must all lie in $D$.

So to decide whether $D$ has $r$ arc-disjoint cycles is equivalent to deciding whether $D'$ contains $r+d_{D'}^+(x)$ arc-disjoint cycles, which is the same as deciding whether $(r+d_{D'}^+(x))$-CPP on $D'$ has a 
solution without the need to duplicate any arc (i.e. the weight of the solution is the same as the sum of all arc-weights in $D'$).
\end{proof}

\medskip 

\noindent{\bf Acknowledgement} We'd like to thank the referee for several useful suggestions for improvement of the paper.

\end{document}